\documentclass[times]{elsarticle}

\usepackage[T1]{fontenc}
\usepackage[utf8]{inputenc}
\usepackage{amsmath}
\usepackage{amssymb}
\usepackage{amsthm}
\usepackage{thmtools}

\declaretheorem[style=remark]{definition}

\declaretheorem[style=remark,sibling=definition]{example}
\declaretheorem[sibling=definition]{theorem}
\declaretheorem[sibling=definition]{lemma}

\declaretheorem[sibling=definition, style=remark]{algorithm}

\usepackage{multicol}
\usepackage{enumitem}

\newlist{steps}{enumerate}{1}
\setlist[steps]{%
  label={\textbf{Step~\arabic*.}},%
  ref={step~\arabic*},%
  leftmargin=1.4142cm}
\setlist[enumerate,1]{label={(\alph*)}}
\setlist[enumerate,2]{label={(\alph{enumi}.\arabic*)}}
\setlist[description,1]{font=\normalfont\itshape}

\usepackage{xcolor}
\definecolor{TodoColour}{HTML}{CE9F6F} 
\definecolor{CSTodoColour}{HTML}{78779B}
\definecolor{MyTodoColour}{HTML}{A5BA93}
\definecolor{LinkColour}{HTML}{044F67}
\definecolor{CiteColour}{HTML}{044F67}
\definecolor{URLColour}{HTML}{044F67}
\definecolor{EmphColour}{HTML}{86ABA5} 
\definecolor{SuperEmphColour}{HTML}{89729E}

\usepackage{hyperref}

\hypersetup{
  final,
  colorlinks, 
  urlcolor=URLColour,
  linkcolor=LinkColour,
  citecolor=CiteColour}

\makeatletter
\AtBeginDocument{\@ifpackageloaded{hyperref}
  {\def\@linkcolor{LinkColour}
   \def\@anchorcolor{blue}
   \def\@citecolor{CiteColour}
   \def\@filecolor{blue}
   \def\@urlcolor{URLColour}
   \def\@menucolor{blue}
   \def\@pagecolor{blue}
\begingroup
  \@makeother\`%
  \@makeother\=%
  \edef\x{%
    \edef\noexpand\x{%
      \endgroup
      \noexpand\toks@{%
        \catcode 96=\noexpand\the\catcode`\noexpand\`\relax
        \catcode 61=\noexpand\the\catcode`\noexpand\=\relax
      }%
    }%
    \noexpand\x
  }%
\x
\@makeother\`
\@makeother\=
}{}}
\makeatother

\def\NoSpace~{{}}

\makeatletter
\let\oldtheequation\theequation
\renewcommand\tagform@[1]{\maketag@@@{\ignorespaces#1\unskip\@@italiccorr}}
\renewcommand\theequation{(\oldtheequation)}
\makeatother

\usepackage[bordercolor=TodoColour, 
 backgroundcolor=TodoColour,  
 linecolor=TodoColour, 
 textsize=todosize]{todonotes}

\definecolor{benimidori}{HTML}{78779B}
\definecolor{koubaiiro}{HTML}{DB5A6B}
\definecolor{QuestionColour}{named}{koubaiiro}

\let\oldemph=\emph
\renewcommand{\emph}[1]{\textcolor{EmphColour}{\oldemph{#1}}}

\DeclareMathOperator{\cont}{ct}
\DeclareMathOperator{\id}{id}

\DeclareMathOperator{\diag}{diag}

\DeclareMathOperator{\supp}{supp}

\renewcommand{\leq}{\leqslant}
\renewcommand{\geq}{\geqslant}
\renewcommand{\kappa}{\varkappa}
\renewcommand{\rho}{\varrho}

\newcommand{\qqtext}[1]{\qquad\text{#1}\qquad}
\newcommand{\qqtextq}[1]{\qquad\text{#1}\quad}
\newcommand{\qqand}{\qqtext{and}}

\newcommand{\qqwhere}{\qqtextq{where}}

\renewcommand{\[}{{[\![}}

\newcommand{\ordinal}[1]{\raisebox{0.7ex}{\scriptsize#1}}

\newcommand{\Mat}[3]{#1^{#2\times#3}}
\newcommand{\MatGr}[2]{\operatorname{GL}_{#2}(#1)}
\newcommand{\CV}[2]{#1^{#2}}

\newcommand{\ID}[1][\relax]{\mathbf{1}\ifx#1\relax\relax\else_{#1}\fi}
\def\ZEROAUX#1,#2;{_{#1\times#2}}
\newcommand{\ZERO}[1][\relax]{\mathbf{0}\ifx#1\relax\relax\else\ZEROAUX#1;\fi}

\newcommand{\Zinfty}{\hat\Z}
\newcommand{\ef}[3]{
  \varepsilon%
  \ifx#1*\relax\else_{#1}\fi%
  \ifx#2*\relax\else^{#2}\fi%
  \ifx#3*\relax\else(#3)\fi}
\newcommand{\efm}[3]{
  \varepsilon%
  \ifx#1*\relax\else_{#1}\fi%
  \ifx#2*\relax\else^{#2}\fi%
  \ifx#3*\relax\else(\!(#3)\!)\fi}
\newcommand{\EF}[3]{
  E%
  \ifx#1*\relax\else_{#1}\fi%
  \ifx#2*\relax\else^{#2}\fi%
  \ifx#3*\relax\else(#3)\fi}
\newcommand{\EFM}[3]{
  E%
  \ifx#1*\relax\else_{#1}\fi%
  \ifx#2*\relax\else^{#2}\fi%
  \ifx#3*\relax\else(\!(#3)\!)\fi}
\newcommand{\val}[1]{v_{#1}}

\newcommand{\ie}{that is}
\newcommand{\eg}{for example}
\newcommand{\wrt}{with respect to}
\newcommand{\UFD}{unique factorization domain}

\newcommand{\A}{\mathbb{A}}
\newcommand{\locA}{\overline{\A}}

\newcommand{\G}{\mathbb{G}}

\newcommand{\K}{\mathbb{K}}

\newcommand{\Q}{\mathbb{Q}}

\newcommand{\Z}{\mathbb{Z}}

\usepackage{latexsym,amscd,url}

         {}
\newcounter{ex}

\begin{document}

\begin{frontmatter}

  \title {A Family of Denominator Bounds for First Order Linear
    Recurrence Systems}

  \author[fsu]{Mark van Hoeij}
  \ead{hoeij@math.fsu.edu}

\author[xlim]{Moulay Barkatou}
\ead{moulay.barkatou@unilim.fr}

  \author[risc]{Johannes Middeke}
  \ead{jmiddeke@risc.jku.at}


  \address[fsu]{Department of Mathematics\\ 
    Florida State University\\  
    Tallahassee, FL 32306, USA}
  
 \address[xlim]{Universit\'e de Limoges, XLIM\\
    123, Av.~A.~Thomas\\
    87060 Limoges cedex, France}

  \address[risc]{Research Institute for Symbolic Computation (RISC)\\
    Johannes Kepler University\\
    Altenbergerstra\ss{}e 69, 4040 Linz, Austria}

\begin{abstract}
For linear recurrence systems, the problem of finding rational solutions
is reduced to the problem of computing polynomial solutions by computing
a content bound or a denominator bound.
There are several bounds in the literature.
The sharpest bound \cite{VHo98}
leads to polynomial solutions of lower degrees,
but as shown in \cite{GheffarAbramov},  
this advantage need not compensate for the time spent on computing that bound.

To strike the best 
balance between sharpness of the bound versus CPU time spent obtaining it,
we will give a family of bounds. The $J$'th member of this family is similar to \cite{AbramovBarkatou1998}
when $J=1$, similar to \cite{VHo98} when $J$ is large, and novel for intermediate values of $J$,
which give the best balance between sharpness and CPU time.  

The setting for our content bounds are systems $\tau(Y) = MY$ where $\tau$
is an automorphism of a \UFD, and $M$ is an invertible matrix with entries in its field of fractions.
This setting includes the shift case, the $q$-shift case, the multi-basic case and others.
We give two versions, a global version, and a version that bounds each entry separately.
\end{abstract}  
\end{frontmatter}

\section{Introduction}

Let $\A$ be a \UFD\ and let $\tau\colon\A \to \A$ be an
automorphism. We denote
the quotient field of $\A$ by $\K$ 
and extend $\tau$ to $\K$.
This paper considers systems of the form
\begin{equation}
  \tag{(\textsc{sys})}
  \tau(Y) = M Y
  \qqwhere
  M \in \MatGr{\K}n.
\end{equation}
The goal is to reduce the problem
of computing \emph{rational solutions} $Y \in \CV{\K}n$ of  (\textsc{sys})
to computing \emph{polynomial solutions} $Z \in \CV{\A}n$ of a related system.

\begin{definition}[Content Bound]\label{def:CB}
We say that $B \in \K$ is a global \emph{content bound} for (\textsc{sys})
if all of its \emph{rational solutions}  $Y \in \CV{\K}n$ are in $B \cdot \CV{\A}n$.
The denominator $d := {\rm den}(B) \in \A \setminus \{0\}$ is then a \emph{denominator bound},
which means all rational solutions are in $\frac1d \cdot \CV{\A}n$. \\
A vector $(B_1,\ldots,B_n)^t \in \K^n$ is a component-wise \emph{content bound}
if all rational solutions are in $B \cdot \CV{\A}n$ where $B = \diag(B_1,\ldots,B_n)$.
\end{definition}

Note that $0$ is a content bound if and only if there are no non-zero rational solutions.
Content bounds and denominator bounds are found in
\cite{Abramov1995},
\cite{AbramovBarkatou1998}, \cite{Barkatou1999},
\cite{AbramovKhmelnov2012}, \cite{CPS:08}, 
\cite{SchneiderMiddeke2017}, or \cite{Middeke2017}. 

If a content bound $B$ is an invertible scalar or matrix, then we can substitute $Y = B  Z$ in
$\tau(Y) = M Y$ obtaining the equivalent system
\begin{math}
  \tau(Z) = \tau(B^{-1}) M B \, Z
\end{math}
for which all rational solutions are in $\CV{\A}n$.
This way, a denominator or a content bound \emph{reduces rational solutions} $Y \in \CV{\K}n$  \emph{to polynomial solutions} $Z \in \CV{\A}n$.
However, as illustrated in \cite{GheffarAbramov}, there is tension between two goals:
(1) we want a bound that can be computed quickly, and (2) want to minimize the degrees of the entries of $Z$.
The goal in this paper is to strike a good balance between these two goals.

We will formulate our bounds in a fairly general setting, see \autoref{sec:notation} below, though the practical utility 
is mainly for cases that have algorithms for polynomial solutions.

\section{Preliminaries}\label{sec:notation}
For a ring $R$ we will use $R^*$ to denote the group of
units in $R$. The set of $m$-by-$n$ matrices with
entries in $R$ will be written as $\Mat{R}mn$. We use $\MatGr{R}n$ for
the set of $n$-by-$n$ invertible matrices over $R$, while
$A^t$ denotes the transpose of $A$.
For $a_1,\ldots,a_n \in R$ let
$\diag(a_1, \ldots, a_n) \in \Mat{R}nn$ denote the corresponding diagonal matrix.

Let $\A$ be a \UFD\ with quotient field $\K$.
Let $(\A, \tau)$ be a \emph{difference ring};
\ie, $\tau\colon\A\to\A$ is an automorphism.
Extending $\tau$ to $\K$ makes
$(\K,\tau)$ a difference field.

\begin{example}\label{ex:A1}
Let $F$ be a field of characteristic 0.
 The main example is
  $\A = F[x]$ with $\tau$ defined by
  $\tau(f(x)) := f(x + 1)$. This is called the \emph{shift case}. Here $\K = F(x)$.
\end{example}

\begin{example}\label{ex:A2}
  Similarly, if $F$ is a field and $q \in F^*$, we can let $\A = F[x]$
  and $\tau(f(x)) := f(qx)$. This is called the \emph{$q$-shift case}.
\end{example}

\begin{example}\label{ex:A3}
  Let $(\G,\tau)$ be a difference ring and let
  $x_1,\ldots,x_s$ be indeterminates over $\G$. Choose units
  $\alpha_1,\ldots,\alpha_s \in \G^*$ and
  $\beta_1,\ldots,\beta_s \in \G$. Let $\A = \G[x_1,\ldots,x_s]$ and
  extend $\tau$ to $\A$ by
  \begin{math}
    \tau(x_j) = \alpha_j x_j + \beta_j.
  \end{math}
  If $\tau|_\G = \id$ is the identity map, then we
  refer to this as the \emph{multi-basic case}.
\end{example}

\begin{definition}[Sharpness]\label{def:sharp}
  Given two content bounds $B, B'$ for the same system
  $\tau(Y) = M Y$, we say that $B$ is \emph{sharper} than $B'$ if it constrains $Y$ to a smaller set
(i.e.  $B \cdot \A^n \subsetneq B' \cdot \A^n$).
\end{definition}

\begin{example}\label{ex:sharp}
For the shift case $\A = \Q[x]$ and $\tau: x \mapsto x+1$ from \autoref{ex:A1},  let
  \begin{equation*}
    M =
    \begin{pmatrix}
      \tfrac{(x+2)^2 (2 x + 1)}{2 (x+1)^2 (x+3)} &
      \tfrac{-(x+2)^2}{2 x (x+1)^2 (x+3)} \\
      \tfrac{-(x+2)^2}{2 (x+1) (x+3)} &  
      \tfrac{(x+2)^2 (2 x + 1)}{2 x (x+1) (x+3)} 
    \end{pmatrix}
    \in \MatGr{\Q(x)}2.
  \end{equation*}
The rational solutions of $\tau(Y) = M Y$ are
  \begin{equation*}
    V =
    \Bigl\{\,
    \begin{pmatrix}
      \tfrac{(x + 1)(c_1 + c_2 x)}{x (x+2)} \\
      \tfrac{(x + 1)(c_1 - c_2 x)}{x + 2}
    \end{pmatrix}
    \;\Bigm|\;
    c_1, c_2 \in \Q
    \,\Bigr\}
    \subseteq \CV{\Q(x)}2.
  \end{equation*}
  Then $V \subset B \cdot  \CV{\A}2 \subsetneq B' \cdot  \CV{\A}2$  where
  \begin{equation*}
   	B = \frac{x+1}{x(x+2)} \ \ \ {\rm and} \ \ \    B' = \frac1{x(x+2)}.
  \end{equation*}
Here $B$ is a sharper content bound than $B'$.
The component-wise bound
  \begin{equation*}
    C :=
    \begin{pmatrix}
      \tfrac{x+1}{x (x+2)} \\
      \tfrac{x+1}{x+2}
    \end{pmatrix}
    \in \CV{\Q(x)}2
  \end{equation*}
  is sharper still since $V \subset {\rm diag}(C) \CV{\A}2 \subsetneq B \cdot  \CV{\A}2$.
\end{example}

Denominator bounds are more common than content bounds
in the literature (see, \eg,  \cite{Abramov1995},  \cite{AbramovBarkatou1998},
\cite{Barkatou1999}, \cite{AbramovKhmelnov2012}, \cite{CPS:08},
\cite{SchneiderMiddeke2017}, or \cite{Middeke2017}).
If $d$ is a denominator bound, then $1/d$ is a content bound.  However, \autoref{ex:sharp} shows
that a sharp global content bound $B$ need not have that form.

\section{The Exponent Function}\label{sec:ef}

Let $p \in \A$ be a prime (= an irreducible polynomial if $\A = F[x]$) and $a \in \A$. The \emph{valuation} of $a$ at $p$ is
\begin{equation*}
  \val{p}(a)
  = \sup \{ j  \mid p^j \text{ divides } a\}.
\end{equation*}
Note that $\val{p}(a) = \infty$ if and only if $a=0$.
We extend $\val{p}: \K \rightarrow \Z \bigcup \{\infty\}$ by defining
$\val{p}(a/b) = \val{p}(a) - \val{p}(b)$ for fractions $a/b \in \K$.
Then
  \begin{equation} \label{rem:val}
    \val{p}(a + b) \geq \min \{ \val{p}(a), \val{p}(b) \}
    \qqand
    \val{p}(a b) = \val{p}(a) + \val{p}(b).
  \end{equation}
  for all $a,b \in \K$.
For a matrix $A$ let $\val{p}(A)$ denote the minimum
valuation of its entries. Then
\begin{equation} \label{vMat}
	\val{p}(A B) \geq \val{p}(A) + \val{p}(B)
\end{equation}
for matrices $A,B$ with matching sizes.

\begin{definition}[Associates and Content] \label{assoc}
Two elements $a_1, a_2 \in \K$ are called \emph{associates}, denoted $a_1 \sim a_2$,  if $a_1 = u a_2$ for some unit $u \in \A^*$.
Just like polynomial contents in Gauss' lemma, the \emph{content} $\cont(A) \in \K$ of a matrix
$A \in \Mat{\K}nm$ is defined up to $\sim$ by the following equivalent properties:
\begin{enumerate}
\item $A$ can be written as $\cont(A)$ times a matrix in $\Mat{\A}nm$ whose entries have gcd 1.
\item $\val{p}( \cont(A) ) = \val{p}(A)$ for all primes $p$.
\item $\cont(A) = g/d$ where $d$ is the least common multiple of the denominators in $A$, and $g$ is the gcd of the entries of $dA$.
\end{enumerate}
\end{definition}

An element $B \in \K$ is a \emph{content-bound} for $\tau(Y) = MY$ if and only if
\emph{$ \val{p}(B) \leq \val{p}(Y) $}
for all solutions $Y \in \K^n$ and all primes $p$ in $\A$.
So \emph{finding $B$ means finding a lower bound for each $\val{p}(Y)$}.

Let
\begin{equation*}
  D = \Bigl\{ a \in \A
  \;\Bigm|\; a \neq 0 \text{ and } \tau^k(a) \sim a  \text{ for some } k \neq 0 \Bigr\}.
\end{equation*}
The fact that $\A$ is a UFD means that every non-zero $a \in \A$ can be written as a product of finitely many primes,
unique up to $\sim$.
This implies that \emph{$a \in D$ if and only if all its prime factors are in $D$}.

We will only compute a lower bound for $\val{p}(Y)$ at 
primes $p \not\in D$.
That results in a content bound up to some factor $a \in D$.
This is sufficient for the main cases including the shift case (then $D = F^*$),
and the $q$-shift case when $q$ is not a root of unity
(then $D = \{ c x^m \, | \, c \in F^*, m \geq 0\}$).

\begin{definition}[Exponent Function]\label{def:ef}
Fix a prime $p \in \A$. If $c \in \K$ then
we define its \emph{exponent function} as: if $c=0$ then $e = \infty$, otherwise $e$
is the function $e: \Z \rightarrow \Z$
with $e(k) = \val{\tau^k(p)}( c )$ for all $k \in \Z$.
\end{definition}

We only use this for primes $p \not\in D$. If $c \neq 0$ then $e$ has finite support and can be represented
with a finite list containing: a lower bound $\ell$ and upper bound $m$ for the support of $e$,
and the numbers $e(k)$ for $k$ from $\ell$ to $m$.

For a system $\tau(Y) = M Y$ we recursively define a matrix
$M_j$ such that \emph{$\tau^j(Y) = M_j Y$}, as follows:
$M_{0} = I$ and
$M_{j+1}  = \tau^j(M) M_j = \tau(M_j) M$.  For $j<0$ we rewrite this as $M_j = \tau^j(M^{-1}) M_{j+1}$. 
Examples include:
\begin{equation*} M_1 = M, \ \ \ \ M_2 = \tau(M) M,  \ \ \ \ M_{-1} = \tau^{-1}(M^{-1}), \ \ \ \ M_{-2} = \tau^{-2}(M^{-1})  \tau^{-1}(M^{-1}). \end{equation*}
After selecting a prime $p \not\in D$, we denote the exponent function of $c_j := \cont(M_j)$ as $e_j: \Z \rightarrow \Z$.

\begin{example}\label{ex:efm} Let $M$ be as in \autoref{ex:sharp}, then $c_1 = \cont(M_1) = x^{-1}(x+1)^{-2}(x+2)^2(x+3)^{-1}$.
The matrix $M_0$ is always $I$ so $c_0 = 1$. From
  \begin{equation*}
    M_{-1} = \tau^{-1}(M^{-1}) =
    \begin{pmatrix}
      \frac{(2 x - 1) x (x + 2)}{2 (x + 1)^2 (x - 1)} &
      \frac{x + 2}{2 (x + 1)^2 (x - 1)} \\
      \frac{x (x + 2)}{2 (x + 1)^2} &
      \frac{(2 x - 1)(x + 2)}{2 (x + 1)^2}
    \end{pmatrix}
  \end{equation*}
  we obtain $c_{-1} = (x-1)^{-1}(x+1)^{-2}(x+2)$.  After selecting $p = x$ we have
  \begin{equation*}
    e_1(k) =
    \begin{cases}
      -1 &\text{if } k = 0\\
      -2 &\text{if } k = 1\\
      \ \,\ 2 &\text{if } k = 2\\
      -1 &\text{if } k = 3\\
      \ \,\ 0 &\text{otherwise}
    \end{cases} \ \ \ \ \ \ \ \ e_0 = 0 \ \ \ \ \ \ \ \ 
    e_{-1}(k) =
      \begin{cases}        
        -1 &\text{if } k = -1 \\
        \ \,\ 0 &\text{if } k =  0 \\
         -2 &\text{if } k = 1 \\
        \ \,\ 1 &\text{if } k = 2 \\
        \ \,\ 0 &\text{otherwise.}
      \end{cases}
  \end{equation*}
\end{example}

\section{The $J$'th global content bound}

\begin{algorithm}[``The global algorithm'': $J$'th global content bound]~
  \begin{description}    
  \item[Input] $M \in \MatGr{\K}n$ and an integer $J \geq 1$.
  \item[Output] $B \in K$ such that $\exists a \in D$ for which $a Y \in B \cdot A^n$ for any rational solution $Y$. In other words, a {\em content bound up to some factor} $a \in D$.
  In the shift-case $a=1$. In the $q$-shift case if $q$ not a root of unity then $a = x^m$ for some $m$ not computed here.
  \newpage
  \item[Procedure]~
    \begin{enumerate}
    \item Compute $M_j$ and $c_j \coloneqq \cont(M_j)$ for $j \in \{-J \ldots J\}$.  \label{Stepa}
    \item \label{stepb} Let ${\cal P}$ be the set of prime factors
    in the denominators of $c_1$ and $c_{-1}$.
    \item \label{stepc} Select one $p \in {\cal P}$ from each \emph{$\tau$-equivalence class},
    where $p_1$ is $\tau$-equivalent to $p_2$ if $\tau^k(p_1) \sim p_2$ for some $k \in \Z$
    (recall $\sim$ from \autoref{assoc}). \\
   Let $\cal{O}$ be the resulting set of primes.
    \item Let $B := 1$. 
    \item  For each $p \in {\cal O} - D$ \label{Stepf}
    \begin{enumerate}
	\item For each $j \in \{-J \ldots J\}$ compute the \emph{exponent-function}
	$e_j: \Z \rightarrow \Z$ of $c_j$ at $p$. Recall that $e_j$ has finite support and $e_j(k) = \val{\tau^k(p)}( c_j )$. \label{e1}
	\item Let $f$ be the output of the \emph{local algorithm}
	in \autoref{sec:local} with input $e_{-J},\ldots,e_{J}$.
	\item If $f = \infty$ then stop and return $B = 0$. Otherwise, $f: \Z \rightarrow \Z$ has finite support and we set
	$B := B \cdot \prod_{k \in \Z}  \tau^k(p)^{f(k)}$.
	\end{enumerate}
    \item Return $B$.
    \end{enumerate}
  \end{description}
\end{algorithm}

The paper \cite{AbramovBarkatou1998} gives a denominator bound that is based
solely on the denominators of $M$ and $M^{-1}$.
That is similar to the above algorithm with $J=1$, and although it can be sharper with $J=1$, see \autoref{ex:alg.local},
its main novelty is when $J > 1$.  Then the local algorithm uses more data, allowing it to construct a sharper bound 
(see the example in \autoref{sec:ex}).
The goal of the local algorithm in \autoref{sec:local} is
to obtain the sharpest
content bound (up to a factor $a \in D$) that can be derived from the exponent-functions $e_{-J},\ldots,e_J$.
In the shift case, that factor $a \in D$ is simply 1.

In the $q$-shift case, if $q$ is a root of unity then $\tau$ has finite order so
$D = \A - \{0\}$ which makes the output trivial. But the root of unity case
is usually excluded.
If $q$ is not a root of unity then $aY  \in B \cdot \CV{F[x]}n$ for some $a  = x^m$ not computed here. Then
the output $B$ restricts rational solutions $Y$ not to $B \cdot \CV{F[x]}n$ but to $B \cdot \CV{F[x, 1/x]}n$.
In the $q$-case, algorithms to bound the degree of polynomial solutions 
can also bound $m$ (just replace $x, q$ with $1/x, 1/q$). So in the $q$-case, finding 
all solutions in $\CV{F[x, 1/x]}n$ is not meaningfully harder than finding 
all solutions in $\CV{F[x]}n$.

In general, $B$ restricts rational solutions to $B \cdot \CV{\locA}{\ n}$ where $\locA := D^{-1} \A \subseteq \K$ is the localization of $\A$ at $D$.
This reduces solutions over $\K$ to solutions over $\locA$.

\section{Local Bounds}\label{sec:local}

Fix one prime $p \not\in D$. A function $f$ is called
a \emph{local content bound} (for $M$ at $p$) if
\begin{equation}\label{DefLCB} \val{ \tau^k(p) }(Y) \geq f(k) \text{   for all solutions } Y \in \K^n \text{  and all  } k \in \Z. \end{equation}
The local algorithm below will compute such $f$  as follow: \autoref{lem:supp} below will provide an initial $f$,
which is then repeatedly improved with  \autoref{lem:ineq}.

For $j \in \{-J,\ldots,J\}$, let $e_j$ be the exponent function of the content $c_j$ of $M_j$.
If $Y$ is a rational solution of $\tau(Y) = MY$ then $\tau^j(Y) = M_j Y$ and from \autoref{vMat} we get
$\val{\tau^{-j}(q)}( Y )  = \val{q}( \tau^j(Y) ) \geq \val{q}(M_j) + \val{q}( Y ) = \val{q}(c_j) + \val{q}( Y )$ for any prime $q$.
For $q = \tau^{k+j}(p)$ we get
\begin{equation} \label{Yval}
	\val{\tau^{k}(p)}( Y ) \geq  \val{q}(c_j) + \val{q}( Y ) = e_j(k+j) + \val{q}( Y ) \geq e_j(k+j) + f(k+j)
\end{equation}
for any local content bound $f$. We have shown:
\begin{lemma}\label{lem:ineq}
Fix some $J>0$.
If $f$ is a local content bound then
$ \val{\tau^{k}(p)}( Y ) \geq  e_j(k+j) + f(k+j)$, so the function $$f_{\rm new}(k) := \max\{ e_j(k+j) + f(k+j) \ \ | \   -J \leq j \leq J\}  \ \ \ \ \  (k \in \Z)$$
is a local content bound as well.
\end{lemma}
Note that $f_{\rm new}(k) \geq f(k)$ since $f_{\rm new}(k)$ is the maximum of set that contains $e_0(k) + f(k) = f(k)$.
The following picture illustrates for $J=2$ how the lemma uses $2J$ neighbors of $f(k)$ to see if the current
lower bound $f(k)$ for $\val{\tau^{k}(p)}( Y )$ can be improved:
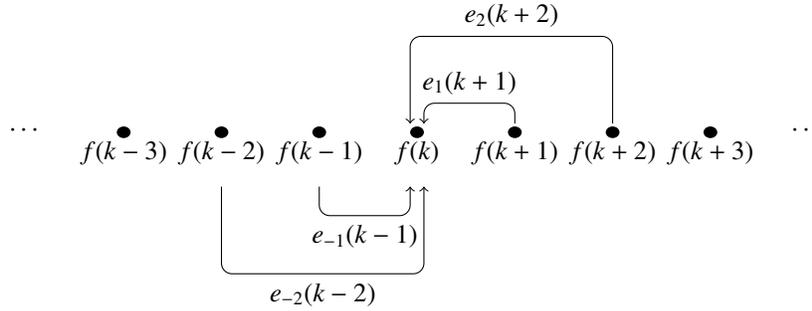
\begin{equation*}
  \begin{tikzpicture}[xscale=1.3, yscale=1.1]
    \fill (-3,0) circle (2pt) node[below] {$f(k-3)$};
    \fill (-2,0) circle (2pt) node[below] {$f(k-2)$};
    \fill (-1,0) circle (2pt) node[below] {$f(k-1)$};
    \fill (0,0) circle (2pt) node[below]  {$f(k)$};
    \fill (1,0) circle (2pt) node[below]  {$f(k+1)$};
    \fill (2,0) circle (2pt) node[below]  {$f(k+2)$};
    \fill (3,0) circle (2pt) node[below]  {$f(k+3)$};
    \node at (-4,0) {$\cdots$};
    \node at (4,0) {$\cdots$};
    \draw[rounded corners, ->] 
    (1,3pt) --
    (1,10pt) -- node[above] {$e_1(k+1)$} 
    (2pt,10pt) -- 
    (2pt,3pt);
    \draw[rounded corners, ->] 
    (2,3pt) --
    (2,33pt) -- node[above] {$e_2(k+2)$} 
    (-2pt,33pt) -- 
    (-2pt,3pt);
    \draw[rounded corners, ->, yshift=-0.2cm] 
    (-1,-13pt) --
    (-1,-23pt) -- node[below] {$e_{-1}(k-1)$} 
    (-2pt,-23pt) -- 
    (-2pt,-13pt);
    \draw[rounded corners, ->, yshift=-0.2cm] 
    (-2,-13pt) --
    (-2,-43pt) -- node[below] {$e_{-2}(k-2)$} 
    (2pt,-43pt) -- 
    (2pt,-13pt);
  \end{tikzpicture}
\end{equation*}

The \emph{support} of $f$ is the set {$\supp(f) = \{ k \in \Z \mid f(k) \neq 0 \}$}.

\begin{lemma}\label{lem:supp}  
Take $\ell_1, m_1, \ell_{-1}, m_{-1} \in \Z$ such
that $\supp(e_1) \subseteq [\ell_1, m_1]$ and  $\supp(e_{-1}) \subseteq [\ell_{-1}, m_{-1}]$.
For every non-zero solution $Y \in \CV{\K}n$ of $\tau(Y) = M Y$, if  $\val{ \tau^k(p) }(Y) \neq 0$ then $k \in [\ell, m]$
  where
  \begin{equation*}
    \ell = \min \{\ell_1, \ell_{-1} + 1\}
    \qqand
    m = \max \{m_1  -1 , m_{-1} \}.
  \end{equation*}
This implies that the function
      $f \colon \Z \to \Z \cup \{-\infty\}$ defined by
      \begin{equation*}
        f(k) =
        \begin{cases}
          -\infty &\text{if } k \in [\ell,m] \\
          \ \ \, \, 0&\text{otherwise.}
        \end{cases}
      \end{equation*}
is a local content bound.
\end{lemma}
\begin{proof}
If there are no non-zero solutions then there is nothing to prove. So
let $Y$ be a generic non-zero solution and let $f(k) = \val{\tau^k(p)}(Y)$.
Recall from \autoref{Yval}\ that $\val{\tau^k(p)}(Y) \geq e_j(k+j) + \val{q}(Y)$ where $q$ was $\tau^{k+j}(p)$, in other words, $$f(k) \geq e_j(k+j) + f(k+j).$$
Since $e_j(k+j) = 0$ when $j=1$ and $k+1 > m_1$ we find $f(k) \geq f(k+1)$ for all $k > m_1 - 1$. For such $k$ we have $f(k) \geq f(k+1) \geq f(k+2) \geq \cdots \geq 0$ since $f$ has finite support. \\
Since $e_j(k+j) = 0$ when $j=-1$ and $k-1 > m_{-1}$ we find $f(k) \geq f(k-1)$ and thus $f(k-1) \leq f(k)$ for all $k -1 > m_{-1}$.  Then 
$f(k) \leq f(k+1) \leq \cdots \leq 0$ for all $k > m_{-1}$. \\[5pt]
Thus $f(k) = 0$ for all $k > m$. The proof for $\ell$ is similar: \\[5pt]
%
%
$f(k) \geq f(k+1)$ for all $k+1 < \ell_1$. Then $0 \geq \cdots \geq f(k-1) \geq f(k)$ for all $k < \ell_1$. \\
$f(k) \geq f(k-1)$ for all $k -1 < \ell_{-1}$.  Then $f(k) \geq 0$ for all $k < \ell_{-1} + 1$.
\end{proof}

\begin{algorithm}[``The local algorithm'':  $J$\ordinal{th} local content bound]\label{alg:local}~
  \begin{description}    
  \item[Input] The exponent-functions $e_{-J},\ldots,e_J$
  from step \autoref{e1} in the global algorithm.
  \item[Output] A local content bound $f\colon \Z \to \Z$ \wrt\ $p$,
    or $\infty$ if it is discovered that there can be no non-zero rational solutions.
  \item[Procedure]~
    \begin{enumerate}[ref=step~(\alph*)]
     \item\label{alg:bounds.2} Let $\ell, m$ and $f$ be as in \autoref{lem:supp}.
    \item\label{alg:bounds.5} Repeat:
      \begin{enumerate}[leftmargin=1.5cm,ref=step~(\alph{enumi}.\arabic*)]
      \item Let $f_{\rm new} \colon \Z \to \Z \cup \{-\infty\}$ be the function given in \autoref{lem:ineq}.
      \item\label{alg:bounds.52} If $f(k) > 0$ for any
        $k \not\in [\ell,m]$ then stop and return $\infty$. \label{b2}
        \item If $f = f_{\rm new}$ then stop and return $f$. \\
        Otherwise set $f := f_{\rm new}$ and Repeat. \label{b3}
      \end{enumerate}
    \end{enumerate}
  \end{description}
\end{algorithm}

\begin{example}\label{ex:alg.local} Let $J=1$ and $p=x$.
\autoref{ex:efm}, which continued \autoref{ex:sharp}, computed
  \begin{equation*}    
    \begin{array}{r|*{9}{r}}
  k   \hspace{3pt}     & \ldots & -2 & -1 & 0  & 1  & \ \ 2 & 3  & \ \ 4 & \ldots \\
      \hline
     e_{-1}(k) & \ldots & 0  & -1 & 0  & -2 & 1 & 0  & 0 & \ldots \\
      e_1(k)  & \ldots & 0  & 0  & -1 & -2 & 2 & -1 & 0 & \ldots
    \end{array}
  \end{equation*}
We do not list $e_0$ since that is always 0.  Then.  
  \begin{equation*}
    \ell_{-1} = -1, \qquad
    \ell_1 = 0, \qqand
    \ell = \min \{\ell_1, \ell_{-1} + 1\} = 0
  \end{equation*}
  and
  \begin{equation*}
    m_{-1} = 2, \qquad
    m_1 = 3, \qqand
    m = \max \{m_1 -1, m_{-1} \} = 2.
  \end{equation*}
  In the algorithm $f: \Z \rightarrow \Z \bigcup \{-\infty\}$ successively becomes
  \begin{equation*}
    \begin{array}{r|*{9}{r}}
      k     \hspace{3pt}     & \ldots & -2 & -1 & 0 & 1  & 2 & 3 & \ \ \, 4 & \ldots \\
      \hline
      f(k) & \ldots & 0 & 0 & -\infty & -\infty & -\infty & 0 & 0 & \ldots \\
      f(k) & \ldots & 0 & 0 & -1 & -\infty & -1 & 0 & 0 & \ldots \\
      f(k) & \ldots & 0 & 0 & -1 & 1 & -1 & 0 & 0 & \ldots
    \end{array}
  \end{equation*}
At that point $f$ stabilizes ($f_{\rm new} = f$) and the local algorithm returns $f$. The global algorithm
converts $f$ to this content bound
\begin{equation*} B = \frac{x+1}{x(x+2)} \end{equation*}
which is sharper than the denominator bound $d = x^2(x+1)(x+2)$ from algorithm UniversalDenominator in Maple, which implements \cite{AbramovBarkatou1998}.
\end{example}

\begin{theorem}\label{thm:alg.local}
 \autoref{alg:local} is correct and terminates.
\end{theorem}
\begin{proof}
Throughout the algorithm $f$ is a local content bound by Lemmas \ref{lem:ineq} and \ref{lem:supp}.
If \autoref{b2} returns $\infty$ then this is correct by \autoref{lem:supp}. Otherwise
the support of $f$ stays inside a finite range $[\ell,m]$. 
As long as $f(k) = -\infty$ for some $k$ we get $f_{\rm new} \neq f$.
So all $f(k)$ are in $\Z$ before the algorithm can terminate in \autoref{b3}.
Since no $f(k)$ ever decreases and the support is bounded, it follows that either (a) the algorithm terminates after finitely many steps, or (b)
some $f(k)$ grows without bound.
Option (b) leads to a contradiction, because if $f(k)$ grows without bound, then so does $f(k+1)$ since $f_{\rm new}(k+1) \geq e_{-1}(k) + f(k)$.
Then $f(k+1), f(k+2),\ldots$ must also grow without bound, which contradicts the fact that the support of $f$ stays inside  $[\ell,m]$.
\end{proof}
The global algorithm only needs to consider primes in $c_1$ or $c_{-1}$, 
otherwise $f$ in \autoref{lem:supp} would be 0.
Correctness of the global algorithm
follows from \autoref{thm:alg.local}.

%


\section{Component-wise Bounds}
We give $\Zinfty := \Z \bigcup \{\infty\} $ the structure of a \emph{tropical semi-ring}
$(\Zinfty, \oplus, \otimes)$ with $\oplus = \min$ and $\otimes = +$. \ \ We extend this to matrices.
If $A  \in \Mat{\Zinfty}mn$ and $B  \in \Mat{\Zinfty}n\ell$ then the $ij$'th entry of $A \otimes B$ is
$$ (A \otimes B)_{ij} \ := \ \bigoplus_{k=1}^n A_{ik} \otimes B_{kj}  \ :=  \ \min\{A_{ik} + B_{kj} \mid 1 \leq k \leq n \}.$$

If $p$ is a prime and $A \in \Mat{\K}mn$ then \emph{$V_p(A) \in \Mat{\Zinfty}mn$} denotes the matrix whose $ij$'th entry is $\val{p}( A_{ij} )$.
The smallest entry is $\val{p}(A)$.
\autoref{rem:val} implies:
  \begin{equation}
  \label{lem:EF.prod}
    V_p(A B) \geq V_p(A) \otimes V_p(B)
  \end{equation}
 for all $A \in \Mat{\K}mn$ and $B \in \Mat{\K}n\ell$, where the inequality is interpreted for each entry separately.

\begin{example}\label{ex:cw}
Let $\A = \Q[x]$, $p=x$ and
  \begin{equation*}
    M =
    \begin{pmatrix}
      -1 + x^3 & -x^2 + x^3 & x \\
      0        & x          & 1 \\
      x + x^2  & x          & 0
    \end{pmatrix}
  \qqand
  Y =
  \begin{pmatrix}
    x \\ -x \\ 1
  \end{pmatrix}.
\end{equation*}
Then
\begin{equation*}
  V_p(M) =
  \begin{pmatrix}
    0      & 2 & 1      \\
    \infty & 1 & 0      \\
    1      & 1 & \infty \\
  \end{pmatrix}
  \qqand
  V_p(Y) = 
  \begin{pmatrix}
    1 \\ 1 \\ 0
  \end{pmatrix}.
\end{equation*}
Lets check \autoref{lem:EF.prod} for $M$ and $Y$:
\begin{equation*}
  \begin{pmatrix}
    3 \\ 0 \\ 3
  \end{pmatrix}
  =
  V_p(\begin{pmatrix} 
      x^3 \\ 1 - x^2 \\ x^3
    \end{pmatrix})
  =
  V_p(M Y)
  \geq
  V_p(M) \otimes V_p(Y)
  =
  \begin{pmatrix}
    \min \{0 + 1, 2 + 1, 1 + 0\} \, \, \\ 
    \min \{\infty + 1, 1 + 1, 0 + 0\} \\ 
    \min \{1 + 1, 1 + 1, \infty + 0\}
  \end{pmatrix}
  =
  \begin{pmatrix}
    1 \\ 0 \\ 2
  \end{pmatrix}.
\end{equation*}
\end{example}

\newpage 

\begin{algorithm}[$J$'th component-wise content bound] \hfill

  \begin{description} 
  \item[Input] $M \in \MatGr{\K}n$ and $J \geq 1$.
  \item[Output] $B \in \K^n$ such that $\exists a \in D$ with $a Y \in \diag(B) A^n$ for any rational solution $Y$.
  \item[Procedure]~
    \begin{enumerate}
    \item Compute $M_j$ for $j \in \{-J \ldots J\}$.
    \item Let ${\cal P}$ be the set of prime factors in the denominators in $M$ and $M_{-1}$.
    \item  $\cal{O}$ := select one $p \in {\cal P}$ from each $\tau$-equivalence class.
    \item Let $B_i := 1$ for $i \in \{1,\ldots,n\}$.
    \item  For each $p \in {\cal O} - D$
    \begin{enumerate}
	\item For $j \in \{-J \ldots J\}$, compute the \emph{exponent-function} $E_j$ of $M_j$ at $p$, which is a function $E_j: \Z \rightarrow \Mat{\Zinfty}nn$ where
	$E_j(k) := V_{\tau^k(p)}(M_j)$.
	\item Call the \emph{local algorithm} below with input $E_{-J} \ldots E_J$.
	\item It returned a function $F: \Z \rightarrow \Zinfty^n$. For $i \in \{1 \ldots n\}$:
	 If $F_i$ (the $i$'th component of $F$) is $\infty$ then $B_i := 0$, otherwise $B_i := B_i \cdot \prod_{k \in \Z}  \tau^k(p)^{F_i(k)}$.
	\end{enumerate}
    \item Return $(B_1,\ldots,B_n)^t$.
    \end{enumerate}
  \end{description}
\end{algorithm}

If an entry of $M_j$ is zero, then the corresponding entry of $E_j(k)$ is $\infty$ for all $k \in \Z$. 
To obtain a finite ``support'', we define \emph{$\supp{(E_j)}$} as the set of all $k \in \Z$ for which $E_j(k) \not\in \Mat{\{0,\infty\}}nn$.
This way we can represent $E_j$ in finite terms with: integers $\ell_j, m_j$ such that $\supp{(E_j)} \subseteq [\ell_j, m_j]$, matrices $E_j(k) \in \Mat{\Zinfty}nn$ for $k \in [\ell_j, m_j]$, and
a matrix we denote as $E_j(\infty) \in \Mat{\{0,\infty\}}nn$ such that $E_j(k) = E_j(\infty)$ for all $k \not\in  [\ell_j, m_j]$.

\begin{algorithm}[$J$\ordinal{th}\ local component-wise content bound]\label{alg:CW.local}~
  \begin{description}
  \item[Input:] $E_{-J},\ldots,E_J$. 
  \item[Output:] $F \colon \Z \to \CV{\Zinfty}n$ such that
    \begin{math}
      F(k) \leq V_{\tau^k(p)}(Y)
    \end{math}
    for all $k \in \Z$ and rational solutions $Y$.
      \item[Procedure:]~
    \begin{enumerate}
    \item
Let $\ell, m$ be as in \autoref{lem:supp},
      let $c=0$  and  let
      $F \colon \Z \to \CV{(\Zinfty \cup \{ - \infty\})}n$
      \begin{equation*}
        F(k) :=
        \begin{cases}
          (-\infty, \ldots, -\infty)^t &\text{if } \ell \leq k \leq m \\
          (0, \ldots, 0)^t &\text{otherwise}.
        \end{cases}
      \end{equation*}
    \item Repeat:
    \begin{enumerate}
    \item  $F_{\rm new}(k) := \max \{ E_j( k+j ) \otimes F(k+j) \mid -J \leq j \leq J \}$  \ \ (for all $k \in \Z)$
    where the maxima are taken component-wise.
    \item If $F_{\rm new} = F$ then stop and return $F$.
    \item \label{c} If all negative entries of $F$ and $F_{\rm new}$ are the same, then $c := c+1$. \\
    If $c > 10$ then return $F_{\rm new}$.
    (For alternatives see \autoref{soph}.)
\item Let $F := F_{\rm new}$ and Repeat.


\end{enumerate}
    \end{enumerate}
  \end{description}
\end{algorithm}

\begin{theorem}\label{thm:alg.CW.local}
  \autoref{alg:CW.local} is correct and terminates.
\end{theorem}
\begin{proof}
As in \autoref{sec:local}, entries can not decrease and the algorithm does not stop if any entries $= -\infty$ remain. Apart from
replacing scalars with matrices and vectors, correctness is proved in the same way as well.
As for termination, negative entries can only increase finitely many times, which makes $c$ in step \autoref{c} a
simple termination mechanism. For more sophisticated versions, see \autoref{soph} below.
\end{proof}

\subsection{Alternatives to an arbitrary cut-off} \label{soph}
The question in this subsection is how to 
ensure termination
without an arbitrary cut-off counter $c$ in step \autoref{c}.
We sketch one approach with an example, and an alternative that is easier to implement.

Let
$M = \left( \begin{array}{cc} x & 0 \\ 0 & 1\end{array} \right)$, take $p = x$, and let $P_n = \tau^{-1}(p)\cdots \tau^{-n}(p) = (x-1)\cdots(x-n)$.
Up to constants, the only rational solution of $\tau(Y) = MY$ is $(0,1)^t$.
Now $(P_n, 1)^t$ is a valid content bound for any $n$ since $Y_1 = 0$ is divisible by any $P_n$.
In every loop, \autoref{alg:CW.local} constructs an $F_{\rm new}$ that is strictly sharper than $F$
(if $F$ encodes $(P_n,1)^t$ then $F_{\rm new}$ encodes $(P_{n+J},1)^t$).
So if we remove step \autoref{c} without implementing an alternative,
then the algorithm will not terminate for $M$.


During the computation $F$ looks as follows. Since $M$ is a 2 by 2 matrix, $F$ has two
components $F_1$ and $F_2$, each of which is a function $\Z \rightarrow \Zinfty \bigcup \{-\infty\}$.
After the first loop $F_2$ is identically 0, while $F_1$ looks like this
$\ldots, 0, 0, 1, \ldots 1, 0, 0, \ldots$ which encodes $P_n$ where $n$ is the number of $1$'s.
This $n$ increases by $J$ in each loop.

We now sketch the first approach to ensure termination without an arbitrary cut-off.
Outside a finite range of $k$'s,
the matrices $E_j(k)$ are constant (recall $E_j(\infty)$ right before \autoref{alg:CW.local}).
If a sufficiently long repeating pattern of positive entries in $F(k)$'s outside of this range
forms during the computation, then, since the $E_j(k)$ are constant here, it is not hard for the
algorithm to prove that this pattern will continue indefinitely.
In the example, when at least $n=1$ positive entries
have formed outside this finite range, then one can immediately deduce from $E_1(\infty)$ that
this pattern can only grow in each loop. But that means that $Y_1$, the first entry of $Y$, must be divisible by a polynomial $P_n$
whose degree keeps increasing. That implies $Y_1 = 0$, so we can replace $F_1$
by the function that is identically $+\infty$.
With this strategy, only finitely many entries $\not\in \{0, \infty\}$ can occur,
because if more than a bounded number appear, the algorithm can construct a proof from $E_{-J}(\infty)\ldots E_{J}(\infty)$
that the pattern will continue, allowing it to replace a component of $F$ by $+\infty$.

We decided not to spell out the details of this approach, because there is a simpler approach which accomplishes a similar outcome.
Let the degree of a rational function be the degree
of the numerator minus the degree of the denominator. To compute rational solutions $Y$, one needs to compute a \emph{degree-bound}
for the entries of $Y$.
For instance, if $Y_1$ is a polynomial of degree $\leq 3$, then the information that $Y_1$ is divisible by $P_4$ is equivalent to
the ``sharper'' bound that $Y_1$ is divisible by $P_{10}$, since both imply $Y_1 = 0$.
So one can design a version of \autoref{alg:CW.local} where the arbitrary cut-off $c>10$ 
is replaced with a cut-off informed by a degree-bound.

Among these alternatives, while the arbitrary cut-off approach is the least elegant, we presented it as the default
because it takes the least amount of implementation effort, and its practical performance, except in very rare cases,
will likely be the same as the alternatives sketched in this subsection.

%
%
%

\section{Example, an eigenring system} \label{sec:ex}
To factor an operator $L = \tau^2 + a_1 \tau + a_0 \in \Q(x)[\tau]$ with the eigenring \cite{eigenring, Barkatou1999} method
we need rational solutions for the system
\begin{equation*} \tau(Y) = MY \ \ \ {\rm where} \ \ \ M =  \left(\begin{array}{cccc} 
0 & 0 & 0 & 1 \\
0 & 0 & -b & -a_1b \\
0 & -a_0 & 0 & -a_1 \\
a_0b & a_0 a_1b & a_1b & a_1^2b \end{array} \right) 
\ \ \ {\rm with} \ \ \ b = \frac1{\tau(a_0)}.
\end{equation*}
For our example\footnote{If an operator $L$ is the LCLM (Least Common Left Multiple) of smaller operators
then $L$ can be factored with the eigenring method.
This example was constructed as LCLM($\tau-x(x+3)/(x+1), \ \tau-(x+1)/x)$. This construction ensures
that $M$ will have at least two (exactly two here) independent rational solutions.} let
\begin{equation*} a_0 = \frac{x^2(x+3)(x^2+5x+5)}{ (x+2)(x-1)(x^2+3x+1) } \ \ \ {\rm and} \ \ \ a_1 = \frac{-(x+1) (x^4+7x^3+11x^2-4x-4)}{(x+2)(x-1)(x^2+3x+1)}.
\end{equation*}
The \emph{global} content bounds for $J \leq 4$ are:
\begin{eqnarray*} B^{\rm global}_{J=1} & = & \frac1{(x-1)x^4(x+1)^3(x+2)(x+3)p q} \\
B^{\rm global}_{J=2} & =&  \frac1{(x-1)x^2(x+1)(x+2)(x+3)p q} \\
B^{\rm global}_{J=3} & =&   \frac1{(x-1)x^2(x+2)(x+3)p q} \\
B^{\rm global}_{J=4} & = & \frac1{(x-1)x^2(x+3)p q} \end{eqnarray*}
where $p = x^2 + 3x + 1$ and $q = \tau(p)$.
The bound from  \cite{AbramovBarkatou1998} (Maple's UniversalDenominator) is the same as $B^{\rm global}_{J=1}$.
Among  \emph{global} content bounds,
$B^{\rm global}_{J=4}$ is sharp (it equals the content of the set of all entries of all rational solutions).
But our component-wise content bounds are sharper still. For $J=1$ and $J=2$ they are:
\begin{eqnarray*}
\left(\frac1{(x-1)x^2(x+2)p},  \frac1{(x^3(x+1)(x+3)q},  \frac1{(x-1)x(x+1)(x+2)p},  \frac1{x(x+1)^2(x+3)q}\right)^t \\
\left(\hspace{16pt} \frac{x+1}{(x-1)p} \hspace{20pt},   \hspace{14pt}       \frac{x+2}{x^2 (x+3)q } \hspace{17pt},   \hspace{30pt} \frac1{(x-1)p}
\hspace{30pt},  \hspace{25pt} \frac{x+2}{x q}
\hspace{23pt} \right)^t \mbox{}
\end{eqnarray*}
The $J=2$ component-wise bound is much sharper than the sharpest global bound. 
In fact, even the $J=1$ component-wise bound is better than the sharpest global bound
(compare the degree of its denominators with that of $B^{\rm global}_{J=4}$).

The component-wise bound for $J = 1$
involves computing $M^{-1}$ but this is done in all variations. For $J=2$ we also have to compute
two matrix products $M_2$ and $M_{-2}$. 
After that, 
we have to compute valuations of their entries, as these valuations form the entries of the exponent-functions $E_j$. If $Q \in \Q(x)$ is an entry of $M_j$, then
a full factorization of $Q$ immediately gives its valuation at every prime $q \in \Q[x]$.
However, a full factorization also computes information we do not need,
since the only valuations we use are at primes $q$ of the form  $\tau^k(p)$ with $p$ as in the algorithm.

We need to compute valuations rapidly in order for the component-wise algorithm to be quick. 
With modular techniques one can quickly compute an upper bound for the valuation of a rational function $Q$ at any $q$,
correctness can then be proved with a trial division.

For special matrices such as the example $M$ above, only a few rational functions need to be factored for the $J=1$ component-wise
algorithm, namely $a_0$ and $a_1$ (then use that $b$ is a shift of $1/a_0$).
The same is true for ``exterior power systems'' which are used \cite{BronsteinTalk} to factor general difference operators
(the eigenring method only factors special cases, in particular LCLM's). The first author's factoring implementation \cite{RFactors} is
set up in a way where rational solutions are already polynomials, but the implementation still computes a component-wise content bound
because it significantly reduces the degrees of the polynomials that the algorithm has to find.

\section*{Acknowledgements}\label{sec:ack}
Mark van Hoeij was supported by NSF grant 1618657.
Johannes Middeke was supported by the Austrian Science Fund (FWF)
grant SFB50 (F5009-N15).

\bibliographystyle{elsarticle-harv}

\providecommand{\noopsort}[1]{}

\end{document}